\documentclass[11pt,a4paper,]{article}
		\usepackage{a4}   

		\usepackage{fancyhdr}
	
		\usepackage[english,ngerman]{babel} 

		\usepackage{xspace}

		\usepackage[fleqn]{amsmath}
		\usepackage{amsthm}
		\usepackage{amssymb}
		\usepackage{amstext}
		\usepackage{amsfonts}
		\usepackage{amsbsy}
		\usepackage{amscd}
		\usepackage{latexsym}
		\usepackage{mathtools}

		\usepackage{textcomp}
	
		\usepackage{graphicx}
	
		\usepackage{tikz}  
		\usetikzlibrary{arrows,positioning, calc, shapes.geometric, shapes.symbols, matrix, fit}            
	
		\usepackage{epsf}
	
		\usepackage{xcolor}
	
		\usepackage{makeidx}
		\makeindex 
	
		\usepackage{cellspace}
	
		\usepackage[utf8]{inputenc}
		
		\usepackage[T1]{fontenc}
	
		\usepackage{lmodern}
	
		\usepackage{listings}
	
		\usepackage{booktabs,longtable,tabularx}
	
		\usepackage{subfig}

		\usepackage{url}

		\usepackage{ifpdf} 

		\usepackage[intoc]{nomencl}

		  \makenomenclature

		\usepackage{enumerate}

		\usepackage[numbers]{natbib}
		\usepackage{bibentry}

		\usepackage{pgfplots}

		\newtheoremstyle{normalstyle}
			{.5em}
			{3pt}
			{\itshape}
			{}
			{\normalfont\bfseries}
			{.\newline}
			{.5em}
			{}
		\theoremstyle{normalstyle}
		\newtheorem{thm}{Theorem}
		\newtheorem{lem}[thm]{Lemma}

		\newtheorem{cor}[thm]{Corollary}
		
		\theoremstyle{definition}

		\usepackage[bookmarks=true,
					colorlinks=true,
					linkcolor=blue
					]
				   {hyperref}

		\lstdefinelanguage{XMLSchema}
					{morekeywords={schema,element,annotation,appinfo,complexType,simpleType,choice,all,sequence},		
			sensitive=true,
			morestring=[b]",
		}
		
		\lstdefinelanguage{ASN1}
			{morekeywords={},		
			sensitive=true,
			morestring=[b]",
		}

		\renewcommand{\epsilon}	{\varepsilon}

	\ifpdf
	  \DeclareGraphicsExtensions{.jpg,.pdf,.png}   
	\else
	  \DeclareGraphicsExtensions{.eps}             
	\fi
	
	\graphicspath{{images/}}

  \makeindex

\newcommand{\comment}[1]{}
\newcommand{\COMMENTOUT}[1]{} 

\newcommand{\N}{\mathbb{N}}
\newcommand{\Z}{\mathbb{Z}}

\newcommand{\R}{\mathbb{R}}

\newcommand{\norm}[1]{\left|\left|#1\right|\right|}

\newcommand{\df}{\mathrel{\mathop:}=}
\newcommand{\fd}{=\mathrel{\mathop:}}

\usepackage[numbers]{natbib}

\usepackage{verbatim}

\title{Faster Algorithms for Integer Programs with Block Structure}
\author{
    Friedrich Eisenbrand\\
    EPFL\\
    friedrich.eisenbrand@epfl.ch
\and
    Christoph Hunkenschr{\"o}der\\
    EPFL\\
    christoph.hunkenschroder@epfl.ch
\and
    Kim-Manuel Klein\\
    EPFL\\
    kim-manuel.klein@epfl.ch
}

\makeindex

\usepackage{anysize}

\usepackage{forest}
\usepackage{relsize}
\usetikzlibrary{trees}
\usetikzlibrary{matrix}
\usepackage{utf8math}

\DeclareMathOperator{\supp}{supp}
\newcommand{\A}{\mathcal{A}}
\newcommand{\T}{\mathcal{T}}
\newcommand{\LO}{\mathcal{O}}

\begin{document}
\selectlanguage{english}

	\maketitle	


\begin{abstract}
We consider integer programming problems  $\max \{ c^Tx : \A  x = b, \, l ≤ x ≤ u, \, x ∈ ℤ^{nt}\}$ where  $\A$ has a  (recursive) block-structure generalizing  \emph{$n$-fold}  integer programs which  recently  received considerable  attention in the literature. An $n$-fold IP is an integer program where $\A$ consists of $n$ repetitions of submatrices $A ∈ ℤ^{r × t}$ on the top horizontal part  and $n$ repetitions of a matrix $B ∈ ℤ^{s × t}$  on the diagonal below the top part.   
Instead of allowing only two types of block matrices, one for the horizontal line and one for the diagonal, we generalize the $n$-fold  setting to allow for arbitrary matrices in every block. We show that such an integer program can be solved in time $n^2t^2 φ ⋅ (r\,s\,Δ)^{\LO(rs^2+ sr^2)}$ (ignoring logarithmic factors). Here $Δ$ is an upper bound on the largest absolute value of an entry of $\A$ and $φ$ is the largest binary encoding length of a coefficient of $c$. This improves upon the previously best algorithm of Hemmecke, Onn and  Romanchuk that runs in time $n^3t^3  φ ⋅ Δ^{\LO(t^2s)}$. In particular, our algorithm is not exponential in the number  $t$ of columns of $A$ and $B$. 

Our algorithm is based on a new upper  bound on the $\ell_1$-norm of an element of the \emph{Graver basis} of an integer matrix and on a proximity bound between the LP and IP optimal solutions tailored for IPs with block structure. These new bounds rely on the \emph{Steinitz Lemma}.

Furthermore,  we extend our techniques to the recently introduced \emph{tree-fold IPs}, where we again present a more efficient algorithm in a generalized setting.  
\end{abstract}

\section{Introduction}
An \emph{integer program (IP)} is an optimization problem of the form 
\begin{equation}
  \label{eq:2}
  \max\{c^Tx : Ax = b, \, l ≤ x ≤ u, \,  x ∈ ℤ^n\}
\end{equation}
which is described by  a  \emph{constraint matrix} $A ∈ ℤ^{m ×n}$, an \emph{objective function} vector $c ∈ ℤ^n$  a \emph{right-hand-side} vector $b ∈ ℤ^m$ and \emph{ lower} and \emph{upper bounds} $l≤x≤u$. Integer programming is one of the most important paradigms in the field of algorithms as a breadth of combinatorial optimization problems 
have an IP-model, see, e.g.~\cite{nemhauser1988integer,schrijver1998theory}. Since integer programming is NP-hard, there is a strong interest in restricted versions of integer programs that can be solved in polynomial time, while still  capturing interesting classes of combinatorial optimization problems. A famous example is the class of integer programs with \emph{totally unimodular} constraint matrix, capturing flow, bipartite matching and shortest path problems for example. This setting has been extended to \emph{bimodular} integer programming recently~\cite{artmann2017strongly}.

Another such polynomial-time solvable restriction is $n$-fold integer programming~\cite{de2008n}. 
Given two matrices $A \in\Z^{r \times t}$ and $B \in\Z^{s \times t}$ and a vector $b \in \Z^{r+ns}$ for some $r,s,t,n \in \Z_+$.
An $n$-fold Integer Program (IP) is an integer program~(\ref{eq:2}) with constraint matrix 
\begin{equation}  \label{eq:classicalnfIP}
    \A =
    \begin{pmatrix}
        A & A & \hdots & A \\
        B & 0 & \hdots & 0\\
        0 & B & & \vdots \\
        \vdots & & \ddots & 0 \\
        0 & \hdots & 0 & B 
      \end{pmatrix}
 \end{equation}
Clearly, one can assume that $t ≥r$ and $t≥s$ holds, as linearly dependent equations can be removed. 
Notice that the number of variables of an $n$-fold integer program is $t ⋅n$. 
 The best known algorithm to solve an $n$-fold IP is due to Hemmecke, Onn, Romanchuk~\cite{hemmecke2013n-fold} with a running time of $\LO(n^3 φ) ⋅ \Delta^{\LO(t(rs+st))}$, where $\Delta$ is the absolute value of the largest entry in $\A$ and $φ$ is the logarithm of the largest absolute value of a component of $c$. 
 For fixed $Δ$, $r$, $s$ and $t$,  the running time depends only polynomially (cubic) on the number of variables and is therefore more efficient than applying  algorithms for general IPs based in lattice-basis reduction~\cite{kannan1987minkowski,lenstra1983integer} or dynamic programming~\cite{papadimitriou1981complexity,eisenbrand2018proximity}.


\bigskip 

\noindent 
The $n$-fold setting has gained strong momentum in the last years, especially  in the fields of \emph{parameterized complexity} and \emph{approximation algorithms}.
An algorithm is \emph{fixed parameter tracktable (fpt)} with respect to a parameter $k$ derived from the input, if its running time is of the form $f(k) \cdot n^{O(1)}$ for some computable function $f$. The result of Hemmecke et al.~\cite{hemmecke2013n-fold} shows that integer programming is fixed parameter tracktable with respect to $Δ,s,r$ and $t$.

This opens the possibility to model combinatorial optimization problems with a fixed parameter as an $n$-fold integer program, see for instance~\cite{knop2017scheduling,chen2017scheduling} and thereby  obtain novel fpt-results.
Very recently Jansen, Klein, Maack and Rau \cite{jansen2018scheduling} used $n$-fold IPs to formulate an enhanced \emph{configuration IP}, that is capable to track additional properties of jobs. With this enhanced IP they were able to develop approximation algorithms for several scheduling problems that involve setups. 
Not only for the scheduling problems, but also in the design of efficient algorithms for string and social choice problems, $n$-fold IPs have been successfully applied~\cite{KnopKM17-bribery,knop2017combinatorial}.

A generalization of the classical $n$-fold IP, called \emph{tree-fold} IP, was very recently introduced by Chen and Marx~\cite{chen2018covering}.
A matrix $A$ is of tree-fold structure,
if it is of recursive $n$-fold structure,
i.e.\ the matrices $B^{(i)}$ in IP~\eqref{eq:classicalnfIP} are of $n^\prime$-fold structure themselves, and so on.
Chen and Marx presented an algorithm to solve tree-fold IPs which runs in time $f(L) \cdot n^3 \varphi$, where $\varphi$ is the encoding length and $L$ involves parameters of the tree like the height of the tree and the number of variables and rows of the involved sub-matrices.
They applied the tree-fold IP to a special case of the traveling salesman problem,
where $m$ clients have to visit every node of a weighted tree and the objective is to minimize the longest tour over all clients.
Using the framework of tree-fold IPs, they obtained an fpt algorithm with a running time of $f(K) \cdot |V|^{\LO(1)}$, where $K$ is the longest tour of a client in the optimal solution and $V$ is the set of vertices of the tree. However, the function $f$ involves a term with a tower of $K$ exponents.

\subsection{Graver Bases and Augmentation Algorithms}
\label{sec:grav-bases-augm}

Before we discuss our contributions, we have to review the core concepts of the algorithm of Hemmecke, Onn and Romanchuk~\cite{hemmecke2013n-fold} in a nutshell.  

Suppose  we are solving a general integer program~(\ref{eq:2}) with constraint matrix $A ∈ ℤ^{m ×n}$ and that  we have a feasible solution $z_0$ at hand. Let $z^*$ be the optimal solution. The vector $z^*-z_0$ lies in the Kernel of $A$, i.e., $A(z^*-z_0)=0$. An integer vector $z ∈ \ker(A)$ is called a \emph{cycle} of $A$.  Two vectors $u,v ∈ ℝ^n$ are said to be \emph{sign compatible} if $u_i ⋅ v_i ≥0$ for each $i$. A cycle $u ∈ \ker(A)$ is \emph{indecomposable} if it is not the sum of two sign-compatible and non-zero cycles of $A$. The set of indecomposable and integral elements from the kernel of $A$ is called the \emph{Graver basis} of $A$,~\cite{graver1975foundations}, see also \cite{onn2010nonlinear, loera2012algebraic}. 

A result of Cook, Fonlupt and Schrijver~\cite{cook1986integer} implies that there exist $2\,n$ Graver-basis elements $g_1,\dots,g_{2\,n} ∈ \ker(A)$ each sign compatible with $z^*-z_0$  such that 
\begin{displaymath}
  z^* - z_0 = ∑_{i=1}^n λ_i g_i 
\end{displaymath}
holds for $λ_i ∈ ℕ_0$. For each $i$ one has that $z_0+ λ_i g_i$ is a feasible integer solution of~(\ref{eq:2}). Furthermore, there exists one $i$ with  $  c^T(z^*-z_0) / (2\, n) ≤ λ_i c^T g_i$. Thus there exists an element $g$ of the Graver basis of  $A$ and a positive integer $λ ∈ ℕ$ such that $z_0 + λ \, g$  is feasible and the gap to the optimum value has been reduced by a factor of $1-1/(2\,n)$. 

Why should it be any simpler to find such an \emph{augmenting vector} $g$ as above? The crucial ingredient that is behind the power of this approach are \emph{bounds} on the $\ell_1$-norm of elements of the Graver basis of $A$. In some cases, these bounds are much more restrictive than the original lower and upper bounds $l≤x≤u$ and thus help in dynamic programming. In fact, each element $g$ of the graver basis of $A$ has $\ell_1$-norm bounded by 
$  \|g\|_1 ≤  δ ⋅ (n -m)$ where $δ$ is the largest  absolute value of a sub-determinant of $A$, see,\cite{onn2010nonlinear}. Applying the Hadamard bound, this means that 
\begin{equation}
  \label{eq:3}
   \|g\|_1 ≤ m^{m/2} Δ^m ⋅ (n-m),
\end{equation}
where $Δ$ is a largest absolute value of an entry of $A$. 
Let us denote $m^{m/2} Δ^m ⋅ (n-m)$ by $G_A$. 
In order to find an augmenting solution which reduces the optimality gap by a factor of (roughly) $1-1/n$ one solves the following \emph{augmentation integer program} with  a suitable $λ$,
\begin{equation}
\label{eq:1}
  \max\{c^Tx : Ax = 0, \, l-z_0 ≤ λ ⋅x ≤ u-z_0, \,  \|x\|_1 ≤  G_A, x∈ ℤ^n\}. 
\end{equation} 
and replaces $z_0$ by $z_0 + λ⋅ x^*$, where $x^*$ is the optimal solution of~(\ref{eq:1}). The number of augmenting steps can be bounded by $\LO(n \log (c^T(z^*-z_0)))$. 

At first sight, it seems that one has not gained much with this
approach, except that the right-hand-side vector $b$
has disappeared. In the case of $n$-fold
integer programming however, the $\ell_1$-norm
of an element of the graver basis of $\A$
is bounded by a function in $r,s,t$
and $Δ$
and thus much smaller than the bound~(\ref{eq:3}). This can be
exploited in dynamic programming approaches.

\subsection*{Contributions of this paper}

We present several elementary observations that, together, result in a much faster algorithm for integer programs in block structure including $n$-fold and tree-fold integer programs. We start with the following. 

\begin{enumerate}[i)]
\item The $\ell_1$-norm of an element of the Graver basis of a given matrix $A ∈ ℤ^{m×n}$ is bounded by $(2\, m ⋅ Δ+1)^m$, where $Δ$ is an upper bound on the absolute value of each entry of $A$. This is shown with the Steinitz lemma and uses similar ideas as in~\cite{eisenbrand2018proximity}. Compared to the previous best bound (\ref{eq:3}), this new bound is independent on the number of columns $n$ of $A$.  \label{item:3}
\end{enumerate}
We then turn our attention to  integer programming problems
\begin{equation}
\label{eq:4}
  \max\{c^Tx : \A\, x = b, \, l ≤ x ≤ u, \,  x ∈ ℤ^{n×t}\}
\end{equation}
 with constraint matrix of the form 
\begin{align*}
    \A =
    \begin{pmatrix}
        A^{(1)} & A^{(2)} & \hdots & A^{(n)} \\
        B^{(1)} & 0 & \hdots & 0\\
        0 & B^{(2)} & & \vdots \\
        \vdots & & \ddots & 0 \\
        0 & \hdots & 0 & B^{(n)} 
    \end{pmatrix}, 
\end{align*}
where $A^{(1)},\dots,A^{(n)} \in \Z^{r \times t}$ and 
$B^{(1)}, \dots, B^{(n)} \in \Z^{s \times t}$ are arbitrary matrices. This is a more general setting than $n$-fold integer programming, since the matrices on the top line and on the diagonal respectively  do not have to repeat. 
In this setting, we obtain the following results. 
\begin{enumerate}[i)]
\setcounter{enumi}{1}
\item The $\ell_1$-norm of an element of the Graver basis of $\A$ is \label{item:4} bounded by $\LO(r\,s\,Δ)^{r\,s}$ which is independent on the number of columns $t$ of the $A^{(i)}$ and $B^{(i)}$. \label{item:1}
\item We next provide a special proximity bound for integer programs with \label{item:2} block structure~(\ref{eq:4}). Let $x^*$ be an arbitrary optimal solution of the linear programming relaxation of~(\ref{eq:4}). We show that there exists an optimal solution $z^*$ of (\ref{eq:4}) with 
  \begin{displaymath}
    \|x^* - z^*\|_1 ≤ n\, t\, (r\,s\,Δ)^{\LO(r\,s)}. 
  \end{displaymath}
\item We then exploit the  bounds \ref{item:1}) and \ref{item:2}) in a new dynamic program to solve~(\ref{eq:4}). Its running time is bounded by 
\[
    n^2t^2 \varphi \log^2 nt \cdot
            (rs\Delta)^{\mathcal{O} (r^2s + rs^2)}
                + \textbf{LP}
\]
where $\varphi$ denotes the logarithm of the largest number
occurring in the input,
and $\textbf{LP}$ denotes the time needed
to solve the LP relaxation of~(\ref{eq:4}).
%
\end{enumerate}

\noindent 
The main advantage of the running time of our algorithm is the improved dependency on the parameter $t$. In contrast, the previous best known algorithm by Hemmecke, Onn and Romanchuk~\cite{hemmecke2013n-fold} for classical $n$-fold IPs invovles a term $\Delta^{\LO(st^2))}$ 
and therefore has an exponential dependency on $t$. Recall that we can assume that $t≥r,s$ holds. The number of columns $t$ can be very large. Even if we do not allow  column-repetitions, $t$ can be as large as $Δ^{r+s}$ and in applications involving \emph{configuration IPs} this is often the order of magnitude one is dealing with. 
Knop, Koutecky and Mnich~\cite{knop2017combinatorial} improved the dependency of $t$  in a special setting  of $n$-fold to  a factor $t^{\LO(r)}$. In their setting, the matrix $B$ on the diagonal consists of one line of ones only. Our running time is an improvement of their result also in this case.




Next, we generalize the notion of tree-fold IPs of \cite{chen2018covering} and allow for arbitrary matrices at each node.
We obtain  the following  natural description of a generalized tree-fold IP:
Given a set of linear equalities $A^{(1)} x = b^{(1)}, \ldots , A^{(N)} x = b^{(N)}$ for some matrices $A^{(i)} \in \Z^{m_i \times n}$ and right hand sides $b^{(i)} \in \Z^{m_i}$ over variables $x \in \Z^{n}$ with upper and lower bounds.
We define a partial order $\preceq$ on the matrices by $A^{(i)} \preceq A^{(j)}$ if the index set of non-zero columns of $A^{(i)}$ is a subset of the index set of non-zero columns of $A^{(j)}$.
Then the IP 
consisting of this set of linear equalities together with an objective function and bounds for the variables $x$ is a \emph{tree-fold IP}, if the partial order $\preceq$ on the matrices $A^{(i)}$ forms a rooted directed tree if arcs stemming from transitivity are omitted. For a precise definition of tree-fold IPs, we refer to Section \ref{sec:tree-fold}.

In this setting we obtain the following result. 

\begin{enumerate}[i)]
\setcounter{enumi}{4}
\item 
We present an algorithm for generalized tree-fold IPs with a running time that depends roughly doubly exponential on the height of the tree (for a precise running time we refer to Lemma \ref{thm:tree-fold-time}). With this algorithm we improve upon the algorithm by Chen and Marx~\cite{chen2018covering}, which has a running time involving a term that has a tower of $\tau$ exponents, where $\tau$ is the height of the tree.

\item 
Using the tree-fold IP formulation of \cite{chen2018covering}, this implies an fpt algorithm for the the traveling salesman problem on trees with $m$ clients with running time $2^{2^{poly(K)}} \cdot |V|^{O(1)}$, where $K$ is the longest tour of an optimal solution over all clients.
\end{enumerate}

\paragraph*{Notation} We use the following notation throughout this paper.
For positive numbers $n,r,s,t \in \N$ and index $i = 1, \dots , n$,
let $A^{(i)} \in \Z^{r \times t}$, $B^{(i)} \in \Z^{s \times t}$
with $\norm{A^{(i)}}_\infty$, $\norm{B^{(i)}}_\infty \leq \Delta$
for some constant $\Delta$.
The $j$-th row of the matrix $A^{(i)}$, $B^{(i)}$ respectively, will be denoted by $A^{(i)}_j$, $B^{(i)}_j$ respectively.
With $\log x$, we denote the logarithm to the basis $2$ of some number $x$.
The logarithm of the largest number occurring in the input is denoted by $\varphi$.



We will often subdivide the set of entries in a vector $y\in \R^{nt}$ or a vector $\A y \in \R^{r + ns}$ into \emph{bricks}.
A vector $y \in \R^{nt}$ will consist of $n$ bricks with $t$ variables each,
i.e.\
\[
    y^T = \begin{pmatrix}
        (y^{(1)})^T, &
        (y^{(2)})^T, &
        \hdots, &
        (y^{(n)})^T
    \end{pmatrix}
\]
with the brick $y^{(i)} \in \R^t$ corresponding to the block $B^{(i)}$.
A vector $g = \A y \in \R^{r + ns}$ will consist of $n+1$ bricks,
\[
    (\A y)^T = \begin{pmatrix}
            (g^{(0)})^T, &
            (g^{(1)})^T, &
            \hdots, &
            (g^{(n)})^T
        \end{pmatrix},
\]
where the first brick $g^{(0)} \in \R^r$ consists of the first $r$ entries
and corresponds to the block row $(A^{(1)},\dots,A^{(n)})$ of $\A$,
and every other block $g^{(i)}$, $i\geq 1$, consists of $s$ entries
and corresponds to the block $B^{(i)}$.
We will always use upper indices with brackets when referring to the bricks,
and the indices will coincide with the index of the block $B^{(i)}$ they correspond to (except brick $g^{(0)}$).
A simple but crucial observation we will use several times is the following.
If $y$ is a cycle of $\A$,
then each brick $y^{(i)}$ is already a cycle of the matrix $B^{(i)}$.

\section{The norm of a  Graver-basis element}
\label{sec:new-bound-length}
In this section, we provide the details of the contributions~\ref{item:3}) and~\ref{item:4}). We will make use of the following lemma of Steinitz~\cite{steinitz1913bedingt,grinberg1980value}. Here  $\|⋅\|$ denotes an arbitrary norm. 

\begin{lem}[Steinitz Lemma]
\label{lem:steinitz}
Let $v_1,\dots,v_n \in \R^m$ be vectors with  
$\norm{v_i} \leq Δ$  for $i=1,\dots,n$.   
If $\sum_{i=1}^m v_i = 0$,  
then there is a reordering $\pi \in S_n$ such that for each $k ∈ \{1,\dots,n\}$ the  partial sum 
$
    p_k \df \sum_{i=1}^k v_{\pi(i)}
$
satisfies 
$\norm{p_k} \leq m \Delta$.
\end{lem}

\begin{lem}
  \label{lem:1}
  Let $A ∈ ℤ^{m ×n}$ be an integer matrix and $Δ$ be an upper bound on the absolute value of each component of $A$ and let $y ∈ℤ^n$ be an element of the Graver basis of $A$. Then $\|y\|_1≤ (2m Δ+1)^m$. 
\end{lem}
\begin{proof}
We define a sequence of vectors $v_1,\dots,v_{\|y\|_1} \in \Z^m$
in the following manner.
If $y_j \geq 0$, we add $y_j$ copies of the $j$-th column of $A$ to the sequence,
if $y_j < 0$ we add $\vert y_j \vert$ copies of the negative
of column $j$ to the sequence.

Clearly, the $v_i$ sum up to zero and their $\ell_∞$-norm is bounded by $Δ$. 
Using Steinitz, there is a reordering $u_1,\dots,u_{\|y\|_1}$ (i.e.\ $v_i = u_{\pi (i)}$ for some permutation $\pi$)
of this sequence s.t.\
each partial sum $p_k \df \sum_{j=1}^k u_j$
is bounded by $m\Delta$ in the $l_\infty$-norm. Clearly 
\begin{displaymath}
  | \{ x ∈ ℤ^n : \|x\|_∞ ≤ m\Delta \} | = \left( 2m \Delta + 1 \right)^m. 
\end{displaymath}
Thus, if $\|y\|_1 > \left( 2m \Delta + 1 \right)^m$, then two of these partial sums are the same and thus $y$ is not indecomposable. This shows the claim. 
\end{proof}

We will now apply the Steinitz lemma to bound the $\ell_1$-norm of an element of the Graver basis of
\begin{displaymath}
 \A =
    \begin{pmatrix}
        A^{(1)} & A^{(2)} & \hdots & A^{(n)} \\
        B^{(1)} & 0 & \hdots & 0\\
        0 & B^{(2)} & & \vdots \\
        \vdots & & \ddots & 0 \\
        0 & \hdots & 0 & B^{(n)} 
    \end{pmatrix}, 
  \end{displaymath}
where $A^{(1)},\dots,A^{(n)} \in \Z^{r \times t}$ and 
$B^{(1)}, \dots, B^{(n)} \in \Z^{s \times t}$ are arbitrary matrices.
Lemma~\ref{lem:1} shows that the $\ell_1$-norm of an element of the graver basis of a matrix $B^{(i)}$ is bounded by $(2\, s Δ+1)^s\fd L_B$.

\begin{lem} \label{lem:cycle_bound_An}
Let $y$ be a Graver-basis element of $\A$, then 
\[
    \norm{y}_1 \leq L_B \left( 2r \Delta L_B + 1 \right)^{r} \fd L_\A.
\]
\end{lem}
\begin{proof}
Let $g$ be a graver basis element of $B^{(i)}$.
Note that as $\norm{g}_1 \leq L_B$ and $\norm{A^{(i)}}_\infty \leq \Delta$,
the vector $A^{(i)} g$ is bounded by
\begin{align}
\label{eq:Ai_g}
    \norm{A^{(i)} g}_\infty &\leq \Delta L_B.
\end{align}
Now consider a graver basis element $y \in \Z^{nt}$ of $\A$
and split it according to the matrices $B^{(i)}$ into bricks,
i.e.\ $y^T = ((y^{(1)})^T, \dots, (y^{(n)})^T)$ with each $y^{(i)} \in \Z^{t}$
being a cycle of $B^{(i)}$.
Hence, each $y^{(i)}$ can be decomposed into the sum of
graver basis elements $y_j^{(i)}$ of $B^{(i)}$ i.e. $y^{(i)} = y^{(i)}_{1} + \ldots + y^{(i)}_{N_i}$ .
Thus, we have a decomposition
\begin{align*}
    0 &= (A^{(1)},\dots,A^{(n)}) y \\
	    &= A^{(1)} y^{(1)} + \dots + A^{(n)} y^{(n)} \\
	    &= A^{(1)} y_1^{(1)} + \dots + A^{(1)} y_{N_1}^{(1)} + \dots
	    	+ A^{(n)} y_1^{(n)} + \dots + A^{(n)} y_{N_n}^{(n)} \\
	    &\fd v_1 + \dots + v_N \in \Z^r
\end{align*}
for some $N = \sum_{i=1}^n N_i$ and $\norm{v_i}_\infty \leq \Delta L_B$
for $i=1,\dots,N$, using~\eqref{eq:Ai_g}.
Now we apply Steinitz to reorder the $v_i$ s.t.\
each partial sum is bounded by $r \Delta L_B$ in the $l_\infty$-norm.
Again, if two partial sums were the same, we could decompose $y$,
thus the number $N$ of summands is bounded by
$\left( 2r \Delta L_B + 1 \right)^{r}$.
Each $v_i$ is the sum of at most $L_B$ columns of some $A^{(j)}$,
hence
\begin{align*}
\norm{y}_1 &\leq L_B \left( 2r \Delta L_B + 1 \right)^{r} \\
	&= \left( 2s \Delta + 1 \right)^s
		\left( 2r \Delta \left( 2s \Delta + 1 \right)^s + 1 \right)^{r} \\
	&= L_A,
\end{align*}
finishing the proof.
\end{proof}

\section{Solving the Generalized n-fold IP}
\label{sec:solving_gen_IP}

Given a feasible solution $x$ of the IP \eqref{eq:4} we
now follow the same principle that we outlined in Section~\ref{sec:grav-bases-augm}. 
There exists an element $y$ of the Graver basis of  $\A$ and a positive integer $λ ∈ ℕ$ such that $x + λ \, y$  is feasible and reducing the gap to the optimum value  by a factor of $1-1/(2\,n)$.
Suppose that we know $λ$. With our bound on $\|y\|_1≤ L_A$ we will find an augmenting vector of at least this quality by solving the following   
\emph{augmentation IP}:
\begin{align}
\label{eq:nfIPker}
    \max c^T y & \\
    \A y &= 0 \nonumber \\
    \norm{y}_1 &\leq L_\A \nonumber \\
    l - z \leq \lambda y &\leq u - z \nonumber \\
    y &\in \Z^{nt} \nonumber
\end{align}
%
The strength of our bound $L_A$ is its independence on $t$. 
We first describe now how to solve this augmentation IP.



\begin{lem}
\label{lem:nfIPker}
Let $\lambda$ be a fixed positive integer.
The augmentation IP ~\eqref{eq:nfIPker},
can be solved in time
$nt \left(r s \Delta \right)^{\mathcal{O}(r^2s+rs^2)}$.
\end{lem}
\begin{proof}
As $\lambda$ is fixed,
it will be convenient to rewrite the bounds on the variables as
\begin{align}
\label{eq:box_bounds}
    l^\star &\leq y \leq u^\star \quad \text{with} \\
    l^\star_i &= \max \left\lbrace \left\lceil
            \frac{l_i - z_i}{\lambda} \right\rceil,
	    		-L_\A \right\rbrace \nonumber \\
    u^\star_i &= \min \left\lbrace \left\lfloor
            \frac{u_i - z_i}{\lambda} \right\rfloor,
	    		 L_\A \right\rbrace . \nonumber
\end{align}
In particular, $u^\star < \infty$.
First observe that for each $y\in \Z^{nt}$ with $\norm{y}_1 \leq L_\A$, one has
\begin{align}
\label{eq:target_space}
    \norm{\A y}_\infty &\leq \Delta L_\A.
\end{align}
We can decompose $y = (y^{(1)}, \dots, y^{(n)})$
into bricks according to the matrices $B^{(i)}$,
and $B^{(k)} y^{(k)} = 0$ has to hold
independently of the other variables.
Let $U\subseteq \Z^{r+s}$ be the set of integer vectors
of infinity norm at most $\Delta L_\A$.
To find an optimal $y^\star$ for the augmentation IP~\eqref{eq:nfIPker}
we construct the following acyclic digraph.
There are two nodes $0_{start}$ and $0_{target}$,
together with $nt$ copies of the set $U$,
arranged in $n$ blocks of $t$ layers as
\[
    U_1^{(1)}, \dots, U_t^{(1)}, U_1^{(2)},
    	\dots, U_t^{(2)}, \dots, U_1^{(n)} \dots U_t^{(n)},
\]
where the $k$-th block will correspond to the matrix
\[
    M^{(k)} \df
    \begin{pmatrix}
        A^{(k)} \\ B^{(k)}
    \end{pmatrix}
\]
(and thus to the brick $y^{(k)}$ of $y$).
Writing $M^{(k)}_j$ for the $j$-th column of the matrix $M^{(k)}$,
the arcs are given as follows.
There is an arc from $0_{start}$ to $v \in U_1^{(1)}$
if there is an integer $y_1$ such that
\[
    v = y_1 M^{(1)}_1
    \quad \text{and} \quad
    l_1^\star \leq y_1 \leq u_1^\star
\]
holds.
The weight of this arc is $c_1 y_1$.

For two nodes $u \in U^{k}_{i-1}$ and
$v \in U^{k}_{i}$ of two consecutive layers in the same block,
we add an arc $(u,v)$ if there is an integer $y_{(k-1)t + i}$ such that
\[
    v-u = y_{(k-1)t + i} M^{(k)}_i
    \quad \text{and} \quad
    l_{(k-1)t + i}^\star \leq y_{(k-1)t + i} \leq u_{(k-1)t + i}^\star
\]
holds, i.e.\ if we can get from $u$ to $v$
by adding the $i$-th column of $\binom{A^{(k)}}{B^{(k)}}$ multiple times.
The weight is $c_{(k-1)t + i} \cdot y_{(k-1)t + i}$.
It remains to define the arcs between two blocks.
If we fix a path through the whole block $U_1^{(k)},\dots,U_t^{(k)}$,
this corresponds to fixing a brick $y^{(k)}$.
Note that $M^{(k)} y^{(k)}$ has to be zero in the last $s$ components,
since continuing with this path in the next block
will not change the entries of $\A y$
corresponding to $B^{(k)}$ any more.
Thus, for placing an arc between two nodes $u \in U^{k}_{t}$ and
$v \in U^{k+1}_{1}$ in two consecutive layers of different blocks,
also the constraints $u_{r+1}=\dots=u_{r+s} = 0$ have to be fulfilled.

Finally, we add arcs from $u \in U_t^{(n)}$ to $0_{target}$
if there exists an integer $y_{nt}$ such that
\[
    -u = y_{nt} M^{(n)}_t
    \quad \text{and} \quad
    l_{nt}^\star \leq y_{nt} \leq u_{nt}^\star
\]
holds.
Again, the weight is $c_{nt} y_{nt}$.

Clearly, a longest $(0_{start}-0_{target})$-path corresponds to
an optimum solution of the augmentation IP~\eqref{eq:nfIPker},
hence it is left to limit the time needed to find such a path.

The out-degree of each node is bounded by
$u_i^\star - l_i^\star \leq 2L_\A+1$ using~\eqref{eq:box_bounds}.
Therefore, the number of arcs is bounded by
\begin{align*}
    nt \cdot \vert U \vert \cdot (2L_\A +1) &=
	    	nt \left(2 \Delta L_\A + 1\right)^{r+s} (2L_\A + 1) \\
    	&\leq nt \left(2 \Delta L_\A + 1\right)^{r+s+1} \\
	    &\leq nt \left( 2 \Delta L_B \left( 2r \Delta L_B + 1 \right)^{r}
		    + 1 \right)^{r+s+1} \\
    	&= nt \cdot \LO (\Delta r)^{r^2s+rs^2 + o (r^2s+rs^2)} \LO (s)^{r^2+rs+r}.
\end{align*}
Using the Moore-Bellman-Ford
algorithm~\cite{aho1974design, korte2012combinatorial}
to find such a path, the claim follows.

To be very precise, we do not necessarily compute an optimum solution of the IP~\eqref{eq:nfIPker}, as the output $y$ might violate $\norm{y}_1 \leq L_\A$. However, what is required in the following lemma is that the output is nonetheless a feasible augmentation step and improves at least as good as an optimum solution of the IP~\eqref{eq:nfIPker}.
\end{proof}
In the following Lemma we consider the value $\Gamma \df \max_i (u_i - l_i)$.
In the case $u < \infty$, we can estimate
$\Gamma \leq 2^{\varphi}$ and obtain a fixed running time in combination with Lemma~\ref{lem:nfIPker}.
However, if there are variables present that are not bounded from above,
we will combine this lemma with the proximity result
of the next Section~\ref{ssec:proximity}
which allows us to introduce artificial upper bounds $u^\prime < \infty$.
\begin{lem}
\label{lem:nfIPred}
Consider the n-fold IP~\eqref{eq:4} with $u < \infty$.
Let $\Gamma \df \max_i (u_i - l_i)$.
Given an initial feasible solution, we can find an optimum solution of the IP
by solving the augmentation IP \eqref{eq:nfIPker}
for a constant $\lambda \in \Z_+$ at most
\[
    \mathcal{O} \left( nt \log (\Gamma) 
            \left( \log (nt\Gamma) + \varphi\right)
        \right)
\]
times, where $\varphi$ is the logarithm of the largest number occurring in the input.
\end{lem}
\begin{proof}
As a first step, we will show that there exists a $\lambda$
such that the improvement we gain from an optimum solution
of the augmentation IP~\eqref{eq:nfIPker} is sufficiently large.

Let $z_0$ be the given feasible solution for the initial IP~\eqref{eq:4},
i.e.\ a vector $z_0 \in \Z^{nt}$ with $\A z_0 = b$,
$l \leq z_0 \leq u$.
Moreover, let $z^\star$ be an optimum integral solution
for the n-fold IP~\eqref{eq:4}.

As $\A (z^\star - z_0) = 0$, a result by Cook, Fonlupt and Schrijver~\cite{cook1986integer} allows us to decompose the vector
\[
    z^\star - z_0 = \sum_{i=1}^{2nt-2} \lambda_i y_i
\]
into an integral conic combination of
sign-compatible Graver basis elements $y_i$ of $\A$.
Note that $z_0 + \lambda_i y_i$ is feasible for each index $i$.
Moreover, by the pigeonhole principle there exists an index $k$ s.t.\
\[
    c^T(\lambda_k y_k) \geq \frac{1}{2nt} c^T (z^\star - z_0),
\]
thus if we add $\lambda_k y_k$ to $z_0$,
we improve at least by a factor of $1/(2nt)$ of the optimum improvement.

However, we do not want to solve~\eqref{eq:nfIPker}
for each possible value of $\lambda$.
Notice that if we replace $\lambda_k$
by $\lambda = 2^{\lfloor \log \lambda_k \rfloor}$,
we only lose a factor of $2$, thus we still improve by
\[
    c^T(\lambda y_k) \geq \frac{1}{4nt} c^T (z^\star - z_0).
\]
Hence, we simply guess $\lambda = 2^i$ for indices $i=0,1,\dots,M$,
compute for each $\lambda$ an optimum solution
$y_\lambda$ of the IP~\eqref{eq:nfIPker},
and pick the best pair $\lambda,y_\lambda$ among all.

It remains to limit the number $M$ of guesses for
$\lambda$ and the number $N$ of iterations.
For limiting $M$,
consider the box constraint
\[
    l - z_0 \leq \lambda y \leq u - z_0.
\]
In each component $i$,
there are only $\lfloor \frac{u_i - l_i}{\lambda}\rfloor + 1$
possible values for $y_i$.
Therefore, if $\lambda \geq \max_i \{u_i-l_i\} \fd \Gamma$,
there is at most one solution.
As we only guess powers of $2$,
the number $M$ of $\lambda$ values we have to guess
is bounded by $\lceil \log \Gamma \rceil + 1$.

We conclude the proof by limiting the number $N$ of augmentation steps.
In each step,
we decrease the distance to an optimal objective function value
by a multiplicative factor of $(4nt-1)/(4nt)$.
(Note that the box constraints ensure that our objective function is bounded.)
Thus $N$ iterations are sufficient whenever
\begin{align*}
    1 &> \left( \frac{4nt-1}{4nt} \right)^N \vert c^T \left( z^\star - z_0 \right) \vert \\
    \Leftrightarrow \quad N &>
            \frac{\log \vert c^T(z^\star - z_0) \vert}
            {\log \left( \frac{4nt}{4nt-1} \right)}
\end{align*}
holds.
Using $\log \frac{4nt}{4nt-1} \geq \frac{1}{4nt} $,
this resolves to the upper bound
\begin{align*}
    \frac{\log \vert c^T(z^\star - z_0) \vert}{\log \left( \frac{4nt}{4nt-1} \right)}
            &\leq 4nt \log \vert c^T(z^\star - z_0) \vert \\
        &\leq 4nt \log \left( nt \max_i \vert c_i \vert \cdot (u_i - l_i)\right) \\
        &\leq 4nt \log \left( nt \Gamma \max_i \vert c_i \vert \right) \\
    	&\in nt \cdot \mathcal{O} \left( \log nt \Gamma  + \varphi \right)
\end{align*}
for $N$.
Thus we have to solve~\eqref{eq:nfIPker} at most
\[
    N \cdot M \in \mathcal{O} \left( nt \log (\Gamma) \left( \log (nt\Gamma) + \varphi\right) \right)
\]
times.
\end{proof}

\subsection{Proximity for $n$-fold IPs}
\label{ssec:proximity}
If no explicit upper bounds are given
(i.e.\ $u_i = \infty$ for some indices $i$),
we cannot bound the number of necessary augmentation steps directly.
To overcome this difficulty,
we will present a proximity result in this section,
stating that for an optimum rational solution $x^\star$,
there exists an optimum integral solution $z^\star$
with $\norm{x^\star - z^\star}_1 \leq nt L_\A$.

With this proximity result,
we can first compute an optimum LP solution $x^\star$,
and then introduce artificial box constraints
$l(x^\star) \leq z \leq u(x^\star)$,
depending on $x^\star$, knowing that at least one optimum IP solution
lies within the introduced bounds.

\begin{lem}
\label{lem:proximity}
Let $x^\star$ be an optimum solution to the LP relaxation of~\eqref{eq:4}.
There exists an optimum integral solution $z^\star$ to~\eqref{eq:4} with
\[
    \norm{x^\star - z^\star}_1
            \leq nt L_\A
            = nt (rs\Delta)^{\LO (rs)}
\]
\end{lem}
\begin{proof}
Let $x^\star$ be an optimum vertex solution
of the LP relaxation of~\eqref{eq:4}
and $z^\star$ be an optimum (integral) solution of~\eqref{eq:4}
that minimizes the $l_1$-distance to $x^\star$.

We say a vector $y$ \emph{dominates} a cycle $y^\prime$
if they are sign-compatible and $\vert y^\prime_i \vert \leq \vert y_i \vert$ for each $i$. The idea is to show that if the $l_1$-distance is too large, we can find a cycle dominated by $z^\star - x^\star$ and either add it to $x^\star$ or subtract it from $z^\star$ leading to a contradiction in both cases.
However, as $z^\star - x^\star$ is fractional, we cannot decompose it directly but have to work around the fractionality.

To this end, denote with $\lfloor x^\star \rceil$ the vector
$x^\star$ rounded towards $z^\star$ i.e. $\lfloor x^\star_i \rceil = \lfloor x^\star_i \rfloor$ if $z^{\star}_i \leq x^{\star}_i$ and $\lfloor x^\star_i \rceil = \lceil x^\star_i \rceil$ otherwise.
Denote with $\{x^\star\}$ the fractional rest i.e.  $\{x^\star\} = x^\star - \lfloor x^\star \rceil $.
Consider the equation
\[
     \A \left( z^\star - x^\star \right) =  \A \left( z^\star - \lfloor x^\star \rceil \right) - \A \{x^\star\} = 0.
\]
Consider the integral vector $\A \{x^\star\}$.
For each index $i$, we will obtain an integral vector $w_i$ out of $\{x^\star\}_i \A_i$ by rounding the entries suitably such that
\[
    \A \{x^\star\} = \sum_{i=1}^{nt} \left( \{x^\star\}_i \A_i \right) = w_1 + \dots + w_{nt}.
\]
To be more formal, fix an index $j$ and let $a_1,\dots,a_{nt}$ denote the $j$-th entry of the vectors $\{x^\star\}_i \A_i$.
Define
\[
f \df \left( \sum_{i=1}^{nt} a_i - \lfloor a_i \rfloor \right) \in \Z_+
\]
as the sum of the fractional parts.
We round up $f$ of the fractional entries $a_i$, and we round down all other fractional entries. If some $a_i$ is integral already, it remains unchanged.
After doing this for each component $j$, we obtain the vectors $w_i$ as claimed.
As $\norm{\{x^\star\}}_\infty \leq 1$, each vector $w_i$ is dominated by either $\A_i$ or $-\A_i$, in particular it inherits the zero entries.

Define the matrix
\[
    \A^\prime \df \left( w_1, \dots , w_{nt} \right).
\]
After permuting the columns, the matrix $(\A,-\A^\prime)$ has $n$-fold structure with parameters $r,s,2t$.
As Lemma~\ref{lem:cycle_bound_An} does not depend on $t$, the Graver basis elements of $(\A,-\A^\prime)$ are bounded by $L_\A$ as well.
We can now identify
\begin{align*}
    \A (z^\star - x^\star) &= \left( \A, - \A^\prime \right)
        \begin{pmatrix}
            z^\star - \lfloor x^\star \rceil \\ \mathbf{1}_{nt}
        \end{pmatrix} = 0,
\end{align*}
and decompose the integral vector $\binom{z^\star - \lfloor x^\star \rceil}{\mathbf{1}_{nt}}$ into Graver basis elements of $l_1$-norm at most $L_\A$.
But if 
\[
    nt L_\A < \norm{z^\star - x^\star}_1 \leq \norm{\begin{pmatrix}
            z^\star - \lfloor x^\star \rceil \\ \mathbf{1}_{nt}
        \end{pmatrix}}_1,
\]
we obtain at least $nt+1$ cycles.
As $\norm{\mathbf{1}_{nt}}_1 = nt$, this grants a cycle $\binom{\bar{y}}{\mathbf{0}_{nt}}$ and hence a cycle $\bar{y}$ of $\A$.

{\bf Case 1:} $c^T \bar{y} \leq 0$: 
As $\bar{y}$ is dominated by $z^\star - \lfloor x^\star \rceil$,
removing cycle $\bar{y}$ from the solution gives a new solution $\bar{z} = z^\star - \bar{y}$ with $c^T \bar{z} \geq c^T z^\star$, which is closer to the fractional solution $x^\star$. However, this contradicts the fact that $z^\star$ was chosen to be a solution with minimal distance $\norm{x^\star - z^\star}_1$.

{\bf Case 2:} $c^T \bar{y} > 0$:
As we rounded $x^\star$ towards $z^\star$ and $\bar{y}$ is dominated by $z^\star - \lfloor x^\star \rceil$,
we can add $\bar{y}$ to $x^\star$ and obtain a better solution,
contradicting its optimality.
\end{proof}

We are now able to state our main theorem regarding the running time of our algorithm to solve a generalized n-foldIP.
\begin{thm} \label{thm:main}
The generalized $n$-fold IP~\eqref{eq:4}
can be solved in time
\[
    n^2t^2 \varphi \log^2 nt \cdot
            \left(rs\Delta\right)^{\mathcal{O} (r^2s + rs^2)}
                + \textbf{LP}
\]
where $\varphi$ denotes the logarithm of the largest number
occurring in the input,
and $\textbf{LP}$ denotes the time needed
to solve the LP relaxation of~\eqref{eq:4}.
%
\end{thm}
\begin{proof}
By Lemma~\ref{lem:nfIPred}, we can solve~\eqref{eq:4}
by solving~\eqref{eq:nfIPker} at most 
\[
    \mathcal{O} \left( nt \log \Gamma 
            \left( \log nt\Gamma + \varphi\right)
        \right)
\]
times, where we can use the bound
$\Gamma \leq 2nt L_\A + 1 = nt (rs\Delta)^{\LO (rs)}$,
after introducing artificial upper bounds $u^\prime$
by Lemma~\ref{lem:proximity},
if necessary.

Given a feasible solution, we can solve~\eqref{eq:nfIPker} in time
\[
    nt \cdot \left(rs\Delta\right)^{\LO (r^2s + rs^2)}
\]
by Lemma~\ref{lem:nfIPker}.
This yields a running time of
\[
    n^2t^2 \varphi \log^2 nt \cdot
        \left(rs\Delta\right)^{\LO (r^2s + rs^2)},
\]
provided a feasible solution.

For finding an initial feasible solution,
we introduce slack variables $y$ and,
for a diagonal matrix $D$ with $D_{ii} = b_i$,
consider the IP
\begin{align*}
    \min \mathbf{1}_{nt}^T y &  \\
    \left(\A , D \right) (x,y)  &= b \nonumber \\
    l \leq x &\leq u \nonumber \\
    x &\in \Z^{nt} \nonumber \\
    y &\geq 0 \nonumber \\
    y &\in \Z^{r + ns}. \nonumber
\end{align*}
Note that $(x,y) = (0,1_{r+ns})$ is a feasible solution.
Moreover, an (optimum) solution with objective function value $0$ corresponds
to a feasible solution for~\eqref{eq:4}.
As we can permute the columns (and the variables accordingly),
this is again an $n$-fold IP, hence we can use our algorithm.
This does not change the running time in terms of Landau symbols.
\end{proof}

\subsection{Separable Convex Objective Functions}
\label{ssec:sep_conv_function}
In this section we consider a separable convex objective function in combination with finite upper bounds $u < \infty$.
A convex function $f: \R^{nt} \rightarrow \R$ is called separable
if there are convex functions $f_i: \R \rightarrow \R$ s.t.\
$f(x) = \sum_{i=1}^{nt} f_i(x_i)$.
Henceforth, we consider the problem
\begin{align}
\label{eq:nfIPconv}
    \min f(x) & \\
    \A x &= b \nonumber \\
    l \leq x &\leq u \nonumber \\
    x &\in \Z^{nt} \nonumber
\end{align}
and show a result similar to Theorem~\ref{thm:main}, see Corollary~\ref{cor:sep_convex_function}.

Let us discuss the differences between a linear function $c^Tx$
and a separable convex function $f(x)$.
First of all, for a linear objective function,
minimizing and maximizing are equivalent,
as one can multiply the objective function by $-1$.
As $-f$ is concave for a convex function $f$,
it will be crucial for this chapter that we want to minimize $f$.

By now, the objective function was used for two things.
In Lemma~\ref{lem:nfIPker}, the constructed graph has weights
according to the objective function.
As the new objective function $f$ is separable,
we have no loss here and a shortest path in the graph
defined with weights according to $f$ is still an optimum solution.
Giving more detail
(and simplifying the notation of the proof of Lemma~\ref{lem:nfIPker}),
let $u$ be a node in the graph we are considering
and $g$ a column we want to add,
the arc $u \rightarrow u + \lambda g$ will have costs
$f_j(z_j + \lambda) - f_j(z_j)$ for the corresponding index $j$ and initial feasible solution $z$ (of the IP~\eqref{eq:4}).
Instead of a longest path, we are now looking for a shortest path.
As the graph is acyclic, this does not make a difference.

More important, we used the objective function
in order to limit the number of iterations.
A key ingredient to do this for $f$ is the following Lemma from
De Loera et alii.

\begin{lem}
[\cite{loera2012algebraic}, Lemma 3.3.1]
\label{lem:deloera}
Let $f(z) \df \sum_{j=1}^k f_j(z_j)$ be separable convex, let $z \in \R^k$,
and let $y_1,\dots,y_l \in \R^k$ be vectors
with the same sign pattern from $\{\leq 0, \geq 0\}^k$;
that is, they belong to a common orthant of $\R^k$.
Then we have
\begin{align}
\label{eq:conv_ineq}
    f \left( z + \sum_{i=1}^l \lambda_i y_i \right) - f(z)
            &\geq \sum_{i=1}^l \lambda_i \left( f(z + y_i) - f(z) \right)
\end{align}
for arbitrary integers $\lambda_1,\dots,\lambda_l \in \Z_+$.
\end{lem}
\begin{proof}
We will show this inequality in one dimension;
the general result follows by separability of $f$.
Let $f: \R \rightarrow \R$ be a convex function and $z \in \R$ a fixed number.
Choose numbers $y_1,\dots,y_l \in \R$ with the same sign.
There exist parameters $\alpha_1,\dots,\alpha_l \in [0,1]$ s.t.\
\begin{alignat*}{2}
    & \quad & (z + y_i) &=
            (1 - \alpha_i) z + \alpha_i \left( z + \sum_{j=1}^l y_j \right) \\
    \Rightarrow & & \sum_{i=1}^l y_i &=
            \sum_{i=1}^l \alpha_i \sum_{j=1}^l y_j \\
    \Rightarrow & & \sum_{i=1}^l \alpha_i &= 1.
\end{alignat*}
Using the convexity of $f$ we find
\begin{alignat*}{2}
    & \quad & f(z + y_i) - f(z) &\leq
            \alpha_i f(z + \sum_{j=1}^l y_j) - \alpha_i f(z) \\
    \Rightarrow & & \left( \sum_{i=1}^l f(z + y_i) - f(z) \right)
            &\leq f(z + \sum_{j=1}^l y_j) - f(z).
\end{alignat*}
The claim follows by writing $\lambda y = \sum_{i=1}^\lambda y$ for any positive integer $\lambda$.
\end{proof}

We are now able to show the following corollary
as the main result of this section.
The running time is similar to the one in Theorem~\ref{thm:main}.
However, we obtain an additional factor that depends logarithmically on the parameter
$M \df \max_{x,y \in P} f(x) - f(y)$, where $P$ is the set of all integral feasible solutions of the IP~\eqref{eq:nfIPconv}.
\begin{cor}
\label{cor:sep_convex_function}
Consider~\eqref{eq:nfIPconv} with finite bounds $l,u < \infty$ and a separable convex function $f$ mapping $\Z^{nt}$ to $\Z$.
Let $P$ the set of feasible integral points for~\eqref{eq:nfIPconv},
and let $M = \max_{x,y \in P} f(x) - f(y)$.
We can solve~\eqref{eq:nfIPconv} in time
\[
    n^2t^2 \varphi \log M (rs \Delta)^{\mathcal{O}(r^2 s + r s^2)}.
\]
\end{cor}
\begin{proof}
We already saw that solving~\eqref{eq:nfIPker} is not affected
by the objective function.
Moreover,
the first step for finding an initial feasible solution is not affected,
as we still do this with optimizing a linear function.
The number of $\lambda$ values for which we
have to solve an IP in each iteration
is still limited by the box constraints,
hence in $\mathcal{O}(\varphi)$.

Thus it remains to limit the number of iterations.
Let $z_0$ be the current solution and
$z^\star = z_0 + \sum_i \lambda_i y_i$ an optimum solution with
Graver basis elements $y_i$.
Multiplying~\eqref{eq:conv_ineq} by $-1$, we get
\[
0 \leq f(z_0) - f ( z^\star )
        \leq \sum_{i=1}^l \left( f(z_0) - f(z_0 + \lambda_i y_i) \right).
\]
Combining this with the result by Cook, Fonlupt and Schrijver~\cite{cook1986integer},
there is a Graver basis element $y$ together with
a multiplicity $\lambda \in \Z_+$ s.t.\
\[
 f(z_0) - f(z_0 + \lambda y)
        \geq \frac{1}{2nt} \left( f(z_0) - f(z^\star) \right).
\]
However, we only solve the type~\eqref{eq:nfIPker} program
for $\lambda$ values that are a power of two.
Hence, we have to ensure that even among those
we find a strong enough improvement.

For this sake let $\widetilde{\lambda} \df 2^{\lfloor \log \lambda \rfloor}$
and choose $1/2 < \gamma \leq 1$ in such a way that
$\widetilde{\lambda} = \gamma \lambda$.
Using convexity once more yields
\begin{align*}
    f(z_0) - f(z_0 + \widetilde{\lambda} y) &\geq 
        f(z_0) - \left[ (1-\gamma)f(z_0) + \gamma f(z_0 + \lambda y) \right] \\
        &= \gamma \left( f(z_0) - f(z_0 + \lambda y)\right) \\
        &\geq \frac12 \left( f(z_0) - f(z_0 + \lambda y)\right) \\
        &\geq \frac{1}{4nt} \left( f(z_0) - f(z^\star) \right).
\end{align*}
As integral vectors are mapped to integral vectors,
we can limit the number of iterations in the same manner as in Lemma~\ref{lem:nfIPred}.
This yields $\mathcal{O}(nt \log M)$ iterations,
where $M = \max_{x,y \in P} (f(x) - f(y))$,
finishing the proof.
\end{proof}

\section{Tree-Fold IPs}
\label{sec:tree-fold}
Given matrices $A^{(i)} \in \Z^{m_i \times n}$ and vectors $b^{(i)} \in \Z^{m_i}$ for $i=1, \ldots , N$ and $c,l,u \in \Z^n$ for some $n, N \in \Z_+$, $m_1,\dots,m_N \in \Z^+$.
We consider the following IP consisting of a system of (systems of) linear equations
\begin{align}
\label{eq:tree-fold-system}
    \max c^T x \\
    A^{(1)} x &= b^{(1)} \nonumber \\
    A^{(2)} x &= b^{(2)} \nonumber \\
        &{\ \vdots} \nonumber \\
    A^{(N)} x &= b^{(N)} \nonumber \\
        l \leq x &\leq u \nonumber\\
        x &\in \Z^n. \nonumber 
\end{align}
Define the support of $A^{(i)}$ as the index set
of all non-zero columns of $A^{(i)}$,
\[
\supp \left( A^{(i)} \right) \df \left\lbrace j \vert \ A^{(i)}_j \neq 0 \right\rbrace,
\]
where $A^{(i)}_j$ denotes the $j$-th column of $A^{(i)}$.
We call~\eqref{eq:tree-fold-system} a tree-fold IP, if the following two conditions hold.
\begin{itemize}
\item For each pair of indices $i,j$ one of the three conditions $\supp ( A^{(i)} ) \subseteq \supp ( A^{(j)} )$, $\supp ( A^{(i)} ) \supseteq \supp ( A^{(j)} )$, or
$\supp ( A^{(i)} ) \cap \supp ( A^{(j)} ) = \emptyset$ is fulfilled.
\item There is an index $k$ s.t.\ for all $i$ we have $\supp ( A^{(i)} ) \subseteq \supp ( A^{(k)} )$.
\end{itemize}

Intuitively, the partial ordering induced by the support forms a tree $T$ on the matrices $A^{(i)}$ (if the arcs stemming from transitivity are omitted). The root of this tree is the matrix with the largest support.
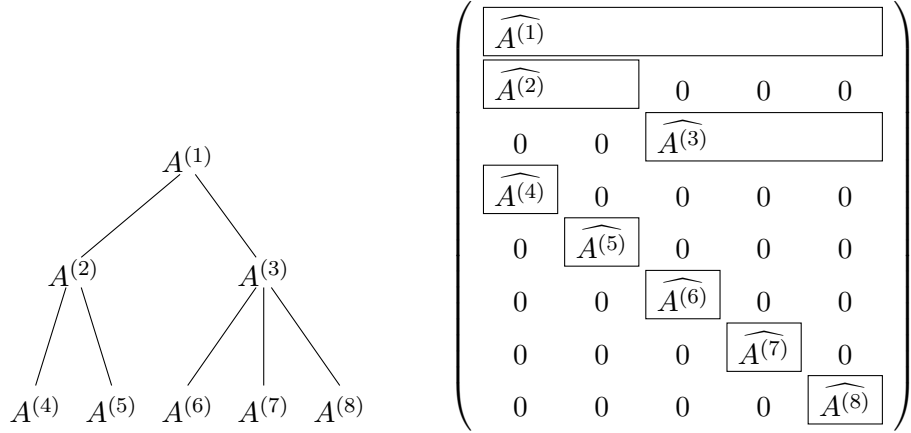
\begin{figure}[hbt]
\begin{center}
\begin{tikzpicture}
	\draw (0,0) node[below]{$A^{(4)}$} -- (0.4,1.3);
	\draw (1,0) node[below]{$A^{(5)}$} -- (0.6,1.3);
	\draw (2,0) node[below]{$A^{(6)}$} -- (2.9,1.3);
	\draw (3,0) node[below]{$A^{(7)}$} -- (3,1.3);
	\draw (4,0) node[below]{$A^{(8)}$} -- (3.1,1.3);
	\draw (0.5,1.5) node{$A^{(2)}$};
	\draw (0.6,1.7) -- (1.9,2.8);
	\draw (3,1.5) node{$A^{(3)}$};
	\draw (2.9,1.7) -- (2.1,2.8);
	\draw (2,3) node{$A^{(1)}$};
\end{tikzpicture}
\hspace{20pt}
\begin{tikzpicture}
\def\b{{\bullet}}
\matrix (m) [matrix of math nodes,
inner sep=3pt, column sep=3pt, row sep=2pt,
nodes={inner sep=0.25em,text width=2em,align=center},
left delimiter=(,right delimiter=),
]{%
    \widehat{A^{(1)}} & \ & \ & \ & \ \\
	\widehat{A^{(2)}} & \ & 0 & 0 & 0 \\
	0 & 0 & \widehat{A^{(3)}} & \ & \ \\
	\widehat{A^{(4)}} & 0 & 0 & 0 & 0 \\
	0 & \widehat{A^{(5)}} & 0 & 0 & 0 \\
	0 & 0 & \widehat{A^{(6)}} & 0 & 0 \\
	0 & 0 & 0 & \widehat{A^{(7)}} & 0 \\
	0 & 0 & 0 & 0 & \widehat{A^{(8)}} \\
};%
\draw (m-1-1.north west) rectangle (m-1-5.south east);
\draw (m-2-1.north west) rectangle (m-2-2.south east);
\draw (m-3-3.north west) rectangle (m-3-5.south east);
\draw (m-4-1.north west) rectangle (m-4-1.south east);
\draw (m-5-2.north west) rectangle (m-5-2.south east);
\draw (m-6-3.north west) rectangle (m-6-3.south east);
\draw (m-7-4.north west) rectangle (m-7-4.south east);
\draw (m-8-5.north west) rectangle (m-8-5.south east);

\end{tikzpicture}
\caption{A matrix tree $T$ and the induced tree-fold matrix $\T$, where $\widehat{A^{(i)}}$ denotes the part of $A^{(i)}$ that consists of non-zero columns.}
\label{fig:matrix_tree}
\end{center}
\end{figure}

Analogously to our $n$-fold results, we will provide an upper bound on the $l_1$-norm of Graver basis elements of tree-fold matrices, together with a proximity result for optimum solutions.
This will be sufficient to obtain an algorithm with a comparable running time.

Throughout this section, $T$ will denote a tree as in Figure~\ref{fig:matrix_tree},
we will denote the depth by $\tau$ and enumerate the layers starting at the deepest leaves
(the leaves are not necessarily all in the same layer).
The whole matrix induced by a tree-fold IP will be denoted by $\T$. This is, the IP~\eqref{eq:tree-fold-system} can be rewritten as
\begin{align}
\label{eq:tree-fold-IP}
    \max c^T x & \\
    \T x &= b \nonumber \\
    l &\leq x \leq u \nonumber \\
    x &\in \Z^{n} \nonumber
\end{align}
\begin{lem}
\label{lem:cycle_bound_tree}
Let $\T$ be a tree-fold matrix
where the corresponding matrix tree $T$ has $\tau$ layers.
Let the matrices of layer $i$ have at most $s_i$ rows and define
$s = \prod_{i=1}^\tau (s_i + 1)$ and $\Delta \df \norm{\T}_\infty$.
Then the Graver basis elements of $\T$ are bounded in their $l_1$-norm by
\[
	L_\tau \leq \left(3 s \Delta \right)^{s-1}.
\]
\end{lem}
\begin{proof}
We enumerate the layers of $T$ starting at the layer with the deepest leaves.
We will prove the claim by induction on the number $\tau$ of layers in the tree $T$.

First observe that for $\tau = 1$, the claim follows by Lemma~\ref{lem:1}, as
\[
    L_1 \leq (2s_1\Delta + 1)^{s_1} \leq (3s\Delta)^{s-1}.
\]
For the induction step, note that every child matrix $A^{(i)}$ of the root in $T$ can be seen as the root matrix of a subtree $T_i$ in $T$ of depth $\tau-1$ with at most $s_1,\dots,s_{\tau-1}$ rows in the corresponding layers.
More formal, delete the root $A^{(1)}$ in $T$ and let $T_i$ be the connected component $A^{(i)}$ is in.
Write
\[
    \widetilde{s} = \prod_{i=1}^{\tau-1} (s_i + 1),
\]
i.e.\ $s = \widetilde{s}(s_{\tau} + 1)$. By induction, we know that all Graver basis elements of the subtree-fold IPs $\T_i$ induced by $T_i$ are bounded by
\begin{align}
\label{eq:subtree_bound}
    L_{\tau -1 } &\leq (3\widetilde{s}\Delta)^{\widetilde{s}-1} \leq (3s\Delta)^{\widetilde{s}-1}.
\end{align}
The rest of the induction step works similar to the proof of Lemma~\ref{lem:cycle_bound_An}. We pick a cycle $y$ of $\T$, decompose it into Graver basis elements for the subtree-fold matrices $\T_i$ and obtain a Steinitz sequence of vectors bounded by the induction hypothesis. 

Denote the root in $T$ by $A^{(1)}$ and let the children of $A^{(1)}$ in $T$ be $A^{(2)}, \dots, A^{(M)}$,
with induced subtree-fold matrices $\T_2, \dots, \T_M$.
This is, we decompose our matrix $\T$ as follows.

\begin{tikzpicture}
\matrix (m) [matrix of math nodes,
inner sep=3pt, column sep=3pt, row sep=2pt,
nodes={inner sep=0.25em,text width=2em,align=center},
left delimiter=(,right delimiter=),
]{%
    A^{(1)} & \ & \ \\
    \ & \ & \ \\
    \T_2 & \cdots & \T_m \\
    \ & \ & \ \\
};%
\draw (m-1-1.north west) rectangle (m-1-3.south east);
\draw (m-2-1.north west) rectangle (m-4-1.south east);
\draw (m-2-2.north west) rectangle (m-4-2.south east);
\draw (m-2-3.north west) rectangle (m-4-3.south east);
\end{tikzpicture}

The submatrices corresponding to the subtrees $T_i$ are denoted with $\T_i$.
Let $y$ be a cycle of $\T$ and partition it into bricks,
\[
    y^T = \left( y^{(2)T}, \dots, y^{(M)T} \right),
\]
such that each $y^{(i)}$ is a cycle of $\T_i$.
Moreover, we split the matrix $A^{(1)} = (A^{(1)}_2,\dots,A^{(1)}_M)$ into blocks according to the supports of the $\T_i$, hence
\[
	A^{(1)} y = A^{(1)}_2 y^{(2)} + \dots + A^{(1)}_M y^{(m)}.
\]
We can decompose the cycles $y^{(i)}$ even further into Graver basis elements of the $\T_i$,
\[
    y^{(i)} = y_1^{(i)} + \dots + y_{m_i}^{(i)}.
\]
This gives us a decomposition of zero,
\begin{align*}
    0 &= A^{(1)} y \\
        &= \sum_{i=2}^M A_i^{(1)} y^{(i)} \\
        &= \sum_{i=2}^M \sum_{j=1}^{m_i} A_i^{(1)} y_{j}^{(i)},
\end{align*}
where by induction hypothesis,
\begin{alignat*}{2}
    & & \norm{y_{j}^{(i)}}_1 &\leq L_{\tau - 1} \\
    \Rightarrow & \quad & \norm{A^{(i)} y_{j}^{(i)}}_\infty &\leq \Delta L_{\tau - 1}.
\end{alignat*}
By Steinitz, there is a reordering $v_1,\dots,v_m$, $m= \sum_i m_i$, of these vectors s.t.\ each partial sum is bounded in the infinity-norm by
$ s_\tau \Delta L_{\tau - 1} $.
If there are two identical partial sums, we can decompose $y$, hence there are at most 
$(2 s_\tau \Delta L_{\tau - 1} + 1)^{s_\tau}$ vectors $A_i^{(1)} y_{j}^{(i)}$.
As each $y_{j}^{(i)}$ has $l_1$-norm at most $L_{\tau -1 }$, we obtain
\begin{align*}
    \norm{y}_1 &\leq L_{\tau - 1} (2 s_\tau \Delta L_{\tau - 1} + 1)^{s_\tau} \\
        &\leq (3 s \Delta)^{\widetilde{s} - 1} (3s \Delta)^{s_\tau \widetilde{s}} \\
        &= (3s\Delta)^{s - 1},
\end{align*}
using~\eqref{eq:subtree_bound}. This finishes the proof.
\end{proof}
In the following Lemma we state a proximity result for tree-fold IPs similar to Lemma~\ref{lem:proximity} for generalized n-fold IPs.
\begin{lem}
\label{lem:proximity_tree}
Let $\T$ be a matrix of tree-fold structure corresponding to the IP~\eqref{eq:tree-fold-system},
and let $x^\star$ be an optimum solution to the LP relaxation of~\eqref{eq:tree-fold-system}.
There exists an optimum integral solution $z^\star$ to~\eqref{eq:tree-fold-system} with
\[
    \norm{x^\star - z^\star}_1  \leq n L_\tau .
\]
\end{lem}
\begin{proof}
The proof works similar to the one of Lemma~\ref{lem:proximity}.
Let $x^\star$ be an optimum vertex solution
of the LP relaxation of~\eqref{eq:tree-fold-IP}
and $z^\star$ be an optimum (integral) solution of~\eqref{eq:tree-fold-IP}
that minimizes the $l_1$-distance to $x^\star$.

The idea is to show that if the $l_1$-distance is too large, we can find a cycle dominated by $z^\star - x^\star$ and either add it to $x^\star$ or subtract it from $z^\star$ leading to a contradiction in both cases.
However, as $z^\star - x^\star$ is fractional, we cannot decompose it directly but have to work around the fractionality.

To this end, denote with $\lfloor x^\star \rceil$ the vector
$x^\star$ rounded towards $z^\star$ i.e. $\lfloor x^\star_i \rceil = \lfloor x^\star_i \rfloor$ if $z^{\star}_i \leq x^{\star}_i$ and $\lfloor x^\star_i \rceil = \lceil x^\star_i \rceil$ otherwise.
Denote with $\{x^\star\}$ the fractional rest i.e.  $\{x^\star\} = x^\star - \lfloor x^\star \rceil $.
Consider the equation
\[
     \T \left( z^\star - x^\star \right) =  \T \left( z^\star - \lfloor x^\star \rceil \right) - \T \{x^\star\} = 0.
\]
with the integral vector $\T \{x^\star\}$.
For each index $i$, consider the fractional column vector $(\{x^\star\})_i (\T_i)$.
We will obtain an integral vector $w_i$ by rounding each vector
$(\{x^\star\})_i (\T_i)$ suitably such that
\[
    \T \{x^\star\} = \sum_{i=1}^n (\{x^\star\})_i (\T_i) = w_1 + \dots + w_{n}.
\]
To be more formal, fix an index $j$ and let $a_1,\dots,a_{n}$ denote the $j$-th entry of the vectors $(\{x^\star\})_i (\T_i)$.
Define
\[
f \df \left( \sum_{i=1}^{n} a_i - \lfloor a_i \rfloor \right) \in \Z_+
\]
as the sum of the fractional parts.
We round up $f$ of the fractional entries $a_i$, and we round down all other fractional entries. If some $a_i$ is integral already, it remains unchanged.
After doing this for each component $j$, we obtain the vectors $w_i$ as claimed.
As $\norm{\{x^\star\}}_\infty \leq 1$, each vector $w_i$ is dominated by either $\T_i$ or $-\T_i$, in particular inhabits the zero entries.

Define the matrix
\[
    \T^\prime \df \left( w_1, \dots , w_{n} \right).
\]
As $\T^\prime$ arises from $\T$ in such a strong way, the matrix $(\T,-\T^\prime)$ has the same tree-fold structure as $\T$ (we basically doubled the number of non-zero columns in each matrix $A^{(i)}$).
As Lemma~\ref{lem:cycle_bound_tree} does not depend on $n$, the Graver basis elements of $(\T,-\T^\prime)$ are bounded by $L_\tau$.
We can now identify
\begin{align*}
    \T (z^\star - x^\star) &= \left( \T, - \T^\prime \right)
        \begin{pmatrix}
            z^\star - \lfloor x^\star \rceil \\ \mathbf{1}_{n}
        \end{pmatrix},
\end{align*}
and decompose the vector $\binom{z^\star - \lfloor x^\star \rceil}{\mathbf{1}_{n}}$ into Graver basis elements dominated by $\binom{z^\star - \lfloor x^\star \rceil}{\mathbf{1}_{n}}$.
But if 
\[
    n L_\T < \norm{z^\star - x^\star}_1 \leq \norm{\begin{pmatrix}
            z^\star - \lfloor x^\star \rceil \\ \mathbf{1}_{n}
        \end{pmatrix}}_1,
\]
we obtain at least $n+1$ cycles.
As $\norm{\mathbf{1}_{n}}_1 = n$, this grants a cycle $\binom{\bar{y}}{\mathbf{0}_{n}}$ and hence a cycle $\bar{y}$ of $\T$.

{\bf Case 1:} $c^T \bar{y} \leq 0$: 
As $\bar{y}$ is dominated by $z^\star - \lfloor x^\star \rceil$,
removing cycle $\bar{y}$ from the solution gives a new solution $\bar{z} = z^\star - \bar{y}$ with $c^T \bar{z} \geq c^T z^\star$, which is closer to the fractional solution $x^\star$. However, this contradicts the fact that $z^\star$ was chosen to be a solution with minimal distance $\norm{x^\star - z^\star}_1$.

{\bf Case 2:} $c^T \bar{y} > 0$:
As we rounded $x^\star$ towards $z^\star$ and $\bar{y}$ is dominated by $z^\star - \lfloor x^\star \rceil$,
we can add $\bar{y}$ to $x^\star$ and obtain a better solution,
contradicting its optimality.
\end{proof}
We conclude with the following theorem that states the running time of our algorithm to solve a tree-fold IP.
\begin{thm}
\label{thm:tree-fold-time}
Let $\T$ be of tree-fold structure with infinity-norm $\Delta$
and corresponding tree $T$.
Let $\tau$ denote the number of layers of $T$
and let the matrices of layer $i$ have at most $s_i$ rows.

Define $s = \prod_{i=1}^\tau (s_i + 1)$ and $\sigma = \sum_{i=1}^\tau s_i$.
Let $n$ denote
the number of columns of $\T$ and $l,u \in (\Z \cup \{\infty\})^n$.
We can solve the IP~\eqref{eq:tree-fold-IP},
\begin{align*}
    \max c^T x & \\
    \T x &= b \nonumber \\
    l &\leq x \leq u \nonumber \\
    x &\in \Z^{n} \nonumber
\end{align*}
in time
\[
n^2 \varphi \log^2 n (s\Delta)^{\mathcal{O}(\sigma s)} + \textbf{LP}
\]
where $\varphi$ denotes the logarithm of the largest number occurring in the input,
and $\textbf{LP}$ denotes the time needed to solve the LP relaxation of \eqref{eq:tree-fold-IP}.
\end{thm}
\begin{proof}
For solving the augmentation IP, we can set up the graph similar to the proof of Lemma~\ref{lem:nfIPker}.
Each layer $U_i$ will correspond to a column of $\T$ and
consist of points $u \in \Z^{m}$
where $m$ is the number of rows in $\T$ with
$\norm{u}_\infty \leq \Delta L_\tau$.
However, as each column of $\A$ corresponds to a path in $T$ starting at the root, each column intersects at most $\tau$ non-zero bricks and has at most $\sigma \df s_1 + \dots + s_\tau$
non-zero entries.
Hence $\vert U_i \vert \leq (2 \Delta L_\tau + 1)^{\sigma}$.
The out-degree is still bounded by $2L_\tau + 1$,
hence the number of arcs is bounded by
\begin{align*}
n (\Delta L_\tau)^\sigma (2 L_\tau + 1)
		&= n (s\Delta)^{\mathcal{O}(\sigma s)},
\end{align*}
yielding this running time.

If $u$ is not finite, we solve the LP relaxation
and use the proximity result in Lemma~\ref{lem:proximity_tree}
in order to replace $l$ and $u$ by finite bounds.
Hence, we have to guess $\LO (s \log (n s \Delta) )$ values for $\lambda$.
The value $c^T (z_{opt} - z_{init})$
can be bounded by $n 2^\varphi (2 L_\tau + 1)$,
hence we have at most
$\LO (n \varphi s \log n s \Delta)$ augmentation steps.
Thus in total, we solve the augmentation IP at most
\[
\LO (n s^2 \varphi \log (ns \Delta)^2)
\]
times.
Finding an initial solution also works analogously, hence the overall running time is
\[
	n^2 \varphi \log^2 n (s\Delta)^{\mathcal{O}(\sigma s)},
\]
finishing the proof.
\end{proof}




		\bibliographystyle{alpha}
		\bibliography{bibliography}

\begin{thebibliography}{DLHOW08}

\bibitem[AH74]{aho1974design}
Alfred~V Aho and John~E Hopcroft.
\newblock {\em The design and analysis of computer algorithms}.
\newblock Pearson Education India, 1974.

\bibitem[AWZ17]{artmann2017strongly}
Stephan Artmann, Robert Weismantel, and Rico Zenklusen.
\newblock A strongly polynomial algorithm for bimodular integer linear
  programming.
\newblock In {\em Proceedings of the 49th Annual ACM SIGACT Symposium on Theory
  of Computing}, pages 1206--1219. ACM, 2017.

\bibitem[CFS86]{cook1986integer}
William Cook, Jean Fonlupt, and Alexander Schrijver.
\newblock An integer analogue of caratheodory's theorem.
\newblock {\em Journal of Combinatorial Theory, Series B}, 40(1):63--70, 1986.

\bibitem[CM18]{chen2018covering}
Lin Chen and Daniel Marx.
\newblock Covering a tree with rooted subtrees--parameterized and approximation
  algorithms.
\newblock In {\em Proceedings of the Twenty-Ninth Annual ACM-SIAM Symposium on
  Discrete Algorithms}, pages 2801--2820. SIAM, 2018.

\bibitem[CMYZ17]{chen2017scheduling}
Lin Chen, D{\'{a}}niel Marx, Deshi Ye, and Guochuan Zhang.
\newblock Parameterized and approximation results for scheduling with a low
  rank processing time matrix.
\newblock In {\em 34th Symposium on Theoretical Aspects of Computer Science,
  {STACS} 2017, March 8-11, 2017, Hannover, Germany}, pages 22:1--22:14, 2017.

\bibitem[DLHK12]{loera2012algebraic}
Jes{\'u}s~A De~Loera, Raymond Hemmecke, and Matthias K{\"o}ppe.
\newblock {\em Algebraic and geometric ideas in the theory of discrete
  optimization}.
\newblock SIAM, 2012.

\bibitem[DLHOW08]{de2008n}
Jes{\'u}s~A De~Loera, Raymond Hemmecke, Shmuel Onn, and Robert Weismantel.
\newblock N-fold integer programming.
\newblock {\em Discrete Optimization}, 5(2):231--241, 2008.

\bibitem[EW18]{eisenbrand2018proximity}
Friedrich Eisenbrand and Robert Weismantel.
\newblock Proximity results and faster algorithms for integer programming using
  the steinitz lemma.
\newblock In {\em Proceedings of the Twenty-Ninth Annual ACM-SIAM Symposium on
  Discrete Algorithms}, pages 808--816. SIAM, 2018.

\bibitem[Gra75]{graver1975foundations}
Jack~E Graver.
\newblock On the foundations of linear and integer linear programming i.
\newblock {\em Mathematical Programming}, 9(1):207--226, 1975.

\bibitem[GS80]{grinberg1980value}
Victor~S Grinberg and Sergey~V Sevast'yanov.
\newblock Value of the steinitz constant.
\newblock {\em Functional Analysis and Its Applications}, 14(2):125--126, 1980.

\bibitem[HOR13]{hemmecke2013n-fold}
Raymond Hemmecke, Shmuel Onn, and Lyubov Romanchuk.
\newblock N-fold integer programming in cubic time.
\newblock {\em Mathematical Programming}, pages 1--17, 2013.

\bibitem[JKMR18]{jansen2018scheduling}
Klaus Jansen, Kim{-}Manuel Klein, Marten Maack, and Malin Rau.
\newblock Empowering the configuration-ip - new {PTAS} results for scheduling
  with setups times.
\newblock {\em CoRR}, abs/1801.06460, 2018.

\bibitem[Kan87]{kannan1987minkowski}
Ravi Kannan.
\newblock Minkowski's convex body theorem and integer programming.
\newblock {\em Mathematics of operations research}, 12(3):415--440, 1987.

\bibitem[KK17]{knop2017scheduling}
Du{\v{s}}an Knop and Martin Kouteck{\`y}.
\newblock Scheduling meets n-fold integer programming.
\newblock {\em Journal of Scheduling}, pages 1--11, 2017.

\bibitem[KKM17a]{knop2017combinatorial}
Dusan Knop, Martin Koutecky, and Matthias Mnich.
\newblock Combinatorial n-fold integer programming and applications.
\newblock In {\em LIPIcs-Leibniz International Proceedings in Informatics},
  volume~87. Schloss Dagstuhl-Leibniz-Zentrum fuer Informatik, 2017.

\bibitem[KKM17b]{KnopKM17-bribery}
Dusan Knop, Martin Kouteck{\'{y}}, and Matthias Mnich.
\newblock Voting and bribing in single-exponential time.
\newblock In {\em 34th Symposium on Theoretical Aspects of Computer Science,
  {STACS} 2017, March 8-11, 2017, Hannover, Germany}, pages 46:1--46:14, 2017.

\bibitem[KV12]{korte2012combinatorial}
Bernhard Korte and Jens Vygen.
\newblock {\em Combinatorial optimization}, volume~2.
\newblock Springer, 2012.

\bibitem[LJ83]{lenstra1983integer}
Hendrik~W Lenstra~Jr.
\newblock Integer programming with a fixed number of variables.
\newblock {\em Mathematics of operations research}, 8(4):538--548, 1983.

\bibitem[NW88]{nemhauser1988integer}
George~L Nemhauser and Laurence~A Wolsey.
\newblock Integer and combinatorial optimization. interscience series in
  discrete mathematics and optimization.
\newblock {\em ed: John Wiley \& Sons}, 1988.

\bibitem[Onn10]{onn2010nonlinear}
Shmuel Onn.
\newblock Nonlinear discrete optimization.
\newblock {\em Zurich Lectures in Advanced Mathematics, European Mathematical
  Society}, 2010.

\bibitem[Pap81]{papadimitriou1981complexity}
Christos~H Papadimitriou.
\newblock On the complexity of integer programming.
\newblock {\em Journal of the ACM (JACM)}, 28(4):765--768, 1981.

\bibitem[Sch98]{schrijver1998theory}
Alexander Schrijver.
\newblock {\em Theory of linear and integer programming}.
\newblock John Wiley \& Sons, 1998.

\bibitem[Ste13]{steinitz1913bedingt}
Ernst Steinitz.
\newblock Bedingt konvergente reihen und konvexe systeme.
\newblock {\em Journal f{\"u}r die reine und angewandte Mathematik},
  143:128--176, 1913.

\end{thebibliography}


\end{document}